\documentclass[1p]{article}

\usepackage{hyperref}

\usepackage{amssymb,amsmath,amsthm,amsfonts}
\usepackage[ruled,vlined]{algorithm2e}
\usepackage{comment,fullpage}

\providecommand{\DontPrintSemicolon}{\dontprintsemicolon}

\newcommand{\innercode}{\ensuremath{\mathcal{C}_0}}
\newcommand{\outercode}{\ensuremath{\mathcal{C}}}
\newcommand{\innerrate}{r_0}
\newcommand{\innerdist}{\delta_0}
\newcommand{\errorrate}{\rho}
\newcommand{\innerquery}{q_0}
\newcommand{\outerquery}{q}
\newcommand{\outerrate}{r}
\newcommand{\smoothness}{s}
\newcommand{\innersmoothness}{s_0}
\newcommand{\updatedquery}{\innerquery'}

\newcommand{\R}{{\mathbb R}}

\newcommand{\F}{{\mathbb F}}

\newcommand{\PR}[1]{{\mathbb{P}}\left\{ #1\right\}}

\newcommand{\twonorm}[1]{\left\|#1\right\|_2}
\newcommand{\onenorm}[1]{\left\|#1\right\|_1}

\newcommand{\tensor}{\otimes}

\newcommand{\ind}[1]{\ensuremath{\mathbf{1}_{#1}}}

\newcommand{\poly}{\mathrm{poly}}

\newcommand{\bigoh}{\mathcal{O}}

\newcommand{\hf}{\frac{1}{2}}

\newcommand{\inset}[1]{\left\{#1\right\}}
\newcommand{\inparen}[1]{\left(#1\right)}

\newcommand{\suchthat}{\,:\,}

\newcommand{\eps}{\varepsilon}

\newcommand{\mkw}[1]{\textcolor{Maroon}{\textbf{(#1 --mary)}}}

\newtheorem{theorem}{Theorem} 
\newtheorem{lemma}[theorem]{Lemma} 
\newtheorem{definition}[theorem]{Definition}

\newtheorem{remark}{Remark}
\newtheorem{claim}[theorem]{Claim}
\newtheorem{proposition}[theorem]{Proposition}

\begin{document}

\title{Local Correctability of Expander Codes}

\author{Brett Hemenway\thanks{Department of Computer and Information Science, University of Pennsylvania, \texttt{fbrett@cis.upenn.edu}.}
\and
Rafail Ostrovsky\thanks{
Department of Computer Science and Department of Mathematics, UCLA, \texttt{rafail@cs.ucla.edu}.
	Research supported
	in part by NSF grants CNS-0830803; CCF-0916574; IIS-1065276; CCF-1016540;
	CNS-1118126; CNS-1136174; US-Israel BSF
	grant 2008411, OKAWA Foundation Research Award, IBM Faculty Research Award,
	Xerox Faculty Research Award, B. John
	Garrick Foundation Award, Teradata Research Award, and Lockheed-Martin
	Corporation Research Award. This material is also
	based upon work supported by the Defense Advanced Research Projects Agency
	through the U.S. Office of Naval Research under
	Contract N00014-11-1-0392. The views expressed are those of the author and
	do not reflect the official policy or position of the 
	Department of Defense or the U.S. Government
}
\and
 Mary Wootters\thanks{Department of Computer Science, Carnegie Mellon University, \texttt{marykw@cs.cmu.edu}.  Research supported in part by NSF grant CCF-1161233}}


\maketitle
\begin{abstract}
In this work, we present the first local-decoding algorithm for expander codes.
This yields a new family of constant-rate codes that can recover from a constant fraction of errors in the codeword symbols, 
and where any symbol of the codeword can be recovered with high probability by reading $N^\eps$ symbols from the corrupted codeword, 
where $N$ is the block-length of the code.

Expander codes, introduced by Sipser and Spielman, are formed from an expander graph $G = (V,E)$ of degree $d$, and an inner code of block-length $d$ 
over an alphabet $\Sigma$.  Each edge of the expander graph is associated with a symbol in $\Sigma$.
A string in $\Sigma^{E}$ will be a codeword if for each vertex in $V$, the symbols on the adjacent edges form a codeword in the inner code.

We show that if the inner code has a smooth reconstruction algorithm in the noiseless setting, then the corresponding expander code has an efficient 
local-correction algorithm in the noisy setting.  Instantiating our construction with inner codes based on finite geometries, 
we obtain novel locally decodable codes with rate approaching one.  This provides an alternative to the multiplicity codes of Kopparty, Saraf and Yekhanin (STOC '11) 
and the lifted codes of Guo, Kopparty and Sudan (ITCS '13).  

\end{abstract}



\section{Introduction}\label{sec:intro}
Expander codes, introduced in ~\cite{SS96}, are linear codes which are notable for their efficient decoding algorithms.  In this paper, we show that when appropriately instantiated, expander codes are also \em locally decodable, \em and we give a sublinear time local-decoding algorithm.

In standard error correction, 
a sender encodes a message $x \in \{0,1\}^k$ as a codeword $c \in \{0,1\}^N$, and 
transmits it to a receiver across a noisy channel.  The receiver's goal is to recover $x$ from the corrupted codeword $w$. 
Decoding algorithms typically process all of $w$ and in turn recover all of $x$.
The goal of local decoding is to recover only a single bit of $x$, with the
benefit of querying only a few bits of $w$.  The number of bits of $w$
needed to recover a single bit $x$ is known as the \emph{query complexity}, and is denoted $\outerquery$.
The important trade-off in local decoding is between query complexity and  the
rate  $\outerrate = k/N$ of the code.  When $\outerquery$ is constant or even
logarithmic in $k$, the best known codes have rates which tend to zero as $N$
grows.
The first locally decodable codes to achieve sublinear locality and rate approaching one were the multiplicity codes of Kopparty, Saraf and Yekhanin \cite{KSY11}.
Prior to this work, only two constructions of locally decodable codes were known with sublinear locality and rate approaching one \cite{KSY11,GKS12}.  
In this paper, we show that expander codes provide a third construction of efficiently locally decodable codes with rate approaching one. 

\subsection{Notation and preliminaries}
Before we state our main results, we set notation and give a few definitions.
We will construct linear codes $\outercode$ of length $N$ and message length $k$, over an alphabet $\Sigma = \F$, for some finite field $\F$.  That is, $\outercode \subset \F^N$ is a linear subspace of dimension $k$.  The \em rate \em of $\outercode$ is the ratio $\outerrate = k/N$.  
We will also use expander graphs: we say a $d$-regular graph $G$ is a \em spectral expander \em with parameter $\lambda$, if $\lambda$ is the second-largest eigenvalue of the normalized adjacency matrix of $G$.  Intuitively, the smaller $\lambda$ is, the better connected $G$ is---see~\cite{HLW06} for a survey of expanders and their applications.
For $n \in \mathbb{Z}$, $[n]$ denotes the set $\{1,2,\ldots, n\}$.
For $x,y \in \Sigma^N$, $\Delta(x,y)$ denotes relative Hamming distance, $x[i]$ denotes the $i^{th}$ symbol of $x$, and $\left. x\right|_S$ denotes $x$ restricted to symbols indexed by $S \subset [N]$. 

A code (along with an encoding algorithm) is \em locally decodable \em if there is an algorithm which can recover a symbol $x[i]$ of the message, making only a few queries to the received word.
\begin{definition}[Locally Decodable Codes (LDCs)]\label{def:ldc}
Let $\mathcal{C} \subset \Sigma^N$ be a code of size $|\Sigma|^k$, and let $E:\Sigma^k \to \Sigma^N$ be an encoding map.  Then $(\mathcal{C},E)$ is $(q,\errorrate)$-locally decodable with error probability $\eta$ if there is a randomized algorithm $R$, so that for any $w \in \Sigma^N$ with $\Delta(w, E(x)) < \errorrate$, for each $i \in [k]$,
\[ \PR{ R(w,i) = x[i] } \geq 1 - \eta,\]
and further $R$ accesses at most $q$ symbols of $w$.  Here, the probability is taken over the internal randomness of the decoding algorithm $R$.
\end{definition}
In this work, we will actually construct \em locally correctable codes, \em which we will see below imply locally decodable codes. 
\begin{definition}[Locally Correctable Codes (LCCs)]\label{def:lcc}
Let $\mathcal{C} \subset \Sigma^N$ be a code, and let $E:\Sigma^k \to \Sigma^N$ be an encoding map.  Then $\mathcal{C}$ is $(q, \errorrate)$-locally correctable with error probability $\eta$ if there is a randomized algorithm, $R$, so that for any $w \in \F^N$ with $\Delta(w,E(x)) < \errorrate$, for each $j \in [N]$,
\[ \PR{ R(w,j) = w[j] } \geq 1 - \eta,\]
and further $R$ accesses at most $q$ symbols of $w$.  Here, the probability is taken over the internal randomness of the decoding algorithm $R$.
\end{definition}

Thus the only difference between locally correctable codes and locally decodable codes is that locally correctable codes can recover symbols of the \emph{codeword} while 
locally decodable codes recover symbols of the \emph{message}.

When there is a constant $\errorrate > 0$ and a failure probability $\eta = o(1)$ so that $\mathcal{C}$ is $(q,\errorrate)$-locally correctable with error probability $\eta$, we will simply say that $\mathcal{C}$ is locally correctable with query complexity $q$ (and similarly for locally decodable).

When $\mathcal{C}$ is a linear code, writing the generator matrix in systematic form gives an encoding function $E:\F^k \to \F^N$ so that for every $x \in \F^k$ and for all $i \in [k]$, $E(x)[i] = x[i]$.  In particular, if $\mathcal{C}$ is a $(q,\errorrate)$ linear LCC, then $(E, \mathcal{C})$ is a $(q,\errorrate)$ LDC.  Because of this connection, we will focus our attention on creating locally correctable linear codes.

Many LCCs work on the following principle: suppose, for each $i \in [N]$, there is a set of $q$ query positions $Q(i)$, which are \em smooth\em---that is, each query is almost uniformly distributed within the codeword---and a method to determine $c[i]$ from $\inset{c[j] \suchthat j \in Q(i)}$ for any uncorrupted codeword $c \in \mathcal{C}$.  If $q$ is constant, 
this \em smooth local reconstruction algorithm \em yields a local correction algorithm: 
with high probability none of the locations queried are corrupted.
In particular, by a union bound, the smooth local reconstruction algorithm is a local correction algorithm that fails with probability at most $\errorrate \cdot q$.
This argument is effective when $q = \bigoh(1)$, however, when $q$ is merely sublinear in $N$, as is the case for us, this reasoning fails.  
This paper demonstrates how to turn codes which only possess a local reconstruction procedure (in the noiseless setting) into LCCs with constant rate and sublinear query complexity. 

\begin{definition}[Smooth reconstruction]\label{def:recon}
For a code $\mathcal{C} \subset \Sigma^N$, consider a pair of algorithms $(Q,A)$, where $Q$ is a randomized query algorithm with inputs in $[N]$ and outputs in $2^N$, and $A:\Sigma^q \times [N] \to \Sigma$ is a deterministic reconstruction algorithm.  
We say that $(Q,A)$ is a \em $\smoothness$-smooth local reconstruction algorithm \em with query complexity $q$ if the following hold.  
\begin{enumerate}
\item For each $i \in [N]$, the query set $Q(i)$ has $|Q(i)| \leq q$.
\item For each $i \in [N]$, there is some set $S \subset [N]$ of size $\smoothness$, so that each query in $Q(i)$ is uniformly distributed in $S$. 
\item For all $i \in [N]$ and for all codewords $c \in \mathcal{C}$, $A( \left.c\right|_{Q(i)},i ) = c[i].$
\end{enumerate}
\end{definition}

If $\smoothness = N$, then we say the reconstruction is perfectly smooth, since all symbols are equally likely to be queried.  Notice that the queries need not be independent.
The codes we consider in this work decode a symbol indexed by $x \in \F^m$ by querying random subspaces through $x$ (but not $x$ itself), and thus will have $\smoothness = N-1$.

\subsection{Related work}

The first local-decoding procedure for an error-correcting code was the
majority-logic decoder for Reed-Muller codes proposed by Reed \cite{R54}.
Local-decoding procedures have found many applications in theoretical computer
science including proof-checking \cite{L90,BFLS91,PS94}, self-testing
\cite{BLR90,BLR93,GLRSW91,GS92} and fault-tolerant circuits \cite{S96a}.  
While these applications implicitly used local-decoding procedures, the first explicit definition of locally decodable
codes did not appear until later \cite{KT00}.  An excellent survey is available \cite{Y10}.
The study of locally decodable codes
focuses on the trade-off between rate (the ratio of message length to codeword
length) and query complexity (the number of queries made by the decoder).
Research in this area is separated into two distinct areas: the first seeks to
minimize the query complexity, while the second seeks to maximize the rate.  In
the low-query-complexity regime, Yekhanin was the first to exhibit codes with a
constant number of queries and a subexponential rate \cite{Y07}.  Following
Yekhanin's work, there has been significant progress in constructing
locally decodable codes with constant query-complexity
\cite{Y08,E09,DGY11,BET10,IS10,CFLWZ11,BIKO12,E12}.  On the other hand, in the
high-rate regime, there has been less progress.  In 2011, Kopparty, Saraf and
Yekhanin introduced \emph{multiplicity codes}, the first codes
with a sublinear local-decoding algorithm \cite{KSY11} and rate approaching one.  
Like Reed-Muller codes, multiplicity codes treat the message as a multivariate polynomial, and create
codewords by evaluating the polynomial at a sequence of points.  Multiplicity
codes are able to improve on the performance of Reed-Muller codes by also
including evaluations of the partial derivatives of the message polynomial in
the codeword.  A separate line of work has developed high-rate locally decodable codes by
``lifting'' shorter codes \cite{GKS12}.  
The work of Guo, Kopparty and Sudan
takes a short code $\innercode$ of length $|\F|^t$, and lifts it to a longer code
$\outercode$, of length $|\F|^m$ for $m > t$ over $\F$, such that every restriction of a codeword in
$\outercode$ to an affine subspace of dimension $t$ yields a codeword in
$\innercode$.  
The definition provides a natural local-correcting
procedure for the outer code: to decode a symbol of the outer code, 
pick a random affine subspace of dimension $t$ that contains the symbol, read
the coordinates and decode the resulting codeword using the code $\innercode$.
Guo, Kopparty and Sudan show how to lift explicit inner codes so that the outer
code has constant rate and query complexity $N^\eps$.  

In this work, 
we show that \em expander codes \em can also give locally decodable codes
with rate approaching one, and
with query complexity $N^\eps$.
Expander codes, introduced by Sipser and Spielman \cite{SS96},
are formed by choosing a $d$-regular expander graph, $G$ on $n$ vertices, and a
code $\innercode$ of length $d$ (called the \emph{inner code}), and defining
the codeword to be all assignments of symbols to the edges of $G$ so that for
every vertex in $G$, its edges form a codeword in $\innercode$.  The connection
between error-correcting codes and graphs was first noticed by Gallager
\cite{G63} who showed that a random bipartite graph induces a good
error-correcting code.  Gallager's construction was refined by Tanner
\cite{tanner81}, who suggested the use of an inner code.  Sipser and Spielman
\cite{SS96} were the first to consider this type of code with an expander
graph, and Spielman \cite{S96} showed that these \emph{expander codes} could be
encoded and decoded in linear time.  Spielman's work provided the first family
of error-correcting code with linear-time encoding and decoding procedures.
The decoding procedure has since been improved by Barg and Zemor
\cite{zemor01,BZ02,BZ05,BZ06}.  

\subsection{Our approach and contributions}
We show that certain expander codes can be efficiently locally decoded, and we instantiate our results to obtain novel families of $(N^\eps, \rho)$-LCCs of rate $1 - \alpha$, for any positive constants $\alpha, \eps$ and some positive constant $\rho$.
Our decoding algorithm runs in time linear in the number of queries, and hence sublinear in the length of the message.
We provide a general method for turning codes with smooth local reconstruction algorithms into LCCs:
our main result, Theorem \ref{thm:tannercorrection}, states that as long as the inner code $\innercode$ has rate at least $1/2$ and possesses a smooth local reconstruction algorithm, then the corresponding family of expander codes are constant rate LCCs.  
In Section \ref{sec:tannerexamples}, we give some examples of appropriate inner codes, leading to the parameters claimed above.

In addition to providing a sublinear time local decoding algorithm for an important family of codes, our constructions are only the third known example of LDCs with rate approaching one, after multiplicity codes~\cite{KSY11} and lifted Reed-Solomon codes~\cite{GKS12}.
Our approach (and the resulting codes) are very different from earlier approaches.  Both multiplicity codes and lifted Reed-Solomon codes use the same basic principle, also at work in Reed-Muller codes: in these schemes, for any two codewords $c_1$ and $c_2$ which differ at index $i$, the corresponding queries $\left. c_1 \right|_{Q(i)}$ and $\left. c_2 \right|_{Q(i)}$ differ in many places.  Thus, if the queries are smooth, with high probability they will not have too many errors, and the correct symbol can be recovered.
In contrast, our decoder works differently: while our queries are smooth, they will not have this distance property.  In fact, changing a mere $\log(q)$ out of our $q$ queries may change the correct answer. 
The trick is that these problematic error patterns must have a lot of structure, and we will show that they are unlikely to occur.

Finally, our results port a typical argument from the low-query regime to the high-rate regime.
As mentioned above, when the query complexity $q$ is constant, a smooth local reconstruction algorithm is sufficient for local correctability.  However, this reasoning fails when $q$ grows with $N$.  In this paper, we show how to make this argument go through: via Theorem \ref{thm:tannercorrection}, any family of codes $\innercode$ with good rate and a smooth local decoder can be used to obtain a family of LCCs with similar parameters.

\section{Local correctability of expander codes}\label{sec:tanner}
In this section, we give an efficient local correction algorithm for expander codes with appropriate inner codes.
We use a formulation of expander codes due to~\cite{zemor01}.
Let $G$ be a $d$-regular expander graph on $n$ vertices with expansion parameter $\lambda$. 
We will take $G$ to be a \em Ramanujan graph, \em  that is, so that $\lambda \leq \frac{2\sqrt{d-1}}{d}$;
explicit constructions of Ramanujan graphs are known~\cite{LPS88,Mar88,Mor94} for arbitrarily large values of $d$.  
Let $H$ be the double cover of $G$.  
That is, $H$ is a bipartite graph whose vertices $V(H)$ are two disjoint copies $V_0$ and $V_1$ of $V(G)$, and so that 
\[ E(H) = \inset{ ( u_0, v_1 ) \suchthat (u,v) \in E(G) },\]
where $u_i$ denotes the copy of $u$ in $V_i$.
Fix a linear inner code $\innercode$ over $\Sigma$ of rate $\innerrate$ and relative distance $\innerdist$.
Let $N = nd$.  For $v_i \in V(H)$, let $E(v_i)$ denote the edges attached to $v$.
The expander code $\outercode \subset \Sigma^{N}$ of length $N$ arising from $G$ and $\innercode$ is given by
\begin{equation}\label{eq:tanner}
\outercode = \outercode_N( \innercode, G) = \inset{ x \in \Sigma^N \suchthat \left. x\right|_{E(v_i)} \in \innercode \text{ for all }  v_i \in V(H) }
\end{equation}
The following theorem states that
as long as the inner code $\innercode$ has good rate and distance, so does the resulting code $\outercode$.
\begin{theorem}[\cite{tanner81,SS96}]
\label{thm:tannercode}
The code $\outercode$ has rate $r \ge  2\innerrate - 1$, and as long as $2 \lambda \leq \innerdist$, the relative distance of $\outercode$ is at least $\innerdist^2/2$.
\end{theorem}

Notice that when $\innerrate < \hf$, Theorem \ref{thm:tannercode} is meaningless.
The rate in Theorem \ref{thm:tannercode} comes from the fact that $\innercode$ has rate $\innerrate$, so each vertex induces $(1-\innerrate)d$ linear constraints, and 
there are $n$ vertices, so the outer code has $nd(1-\innerrate)$ constraints.  Since the outer code has length $N = nd/2$, its rate is at least $2\innerrate - 1$.
This na\"{\i}ve lower bound on the rate ignores the possibility that the constraints induced by the different vertices may not all be independent.
It is an interesting question whether for certain inner codes, a more careful counting of constraints could yield a better lower bound on the rate.
The ability to use inner codes of rate less than $\hf$ would permit much more flexibility in the choice of inner code in our constructions.

The difficulty of a more sophisticated lower bound on the rate was noticed by Tanner, who pointed out that simply permuting the codewords associated with a given 
vertex could drastically alter the parameters of the outer code \cite{tanner81}.

\subsection{Local Correction}\label{sec:tannerdecoding}
If the inner code $\innercode$ has a smooth local reconstruction procedure, then not only does $\outercode$ have good distance, but 
we show it can also be efficiently locally corrected.
Our main result is the following theorem.
\begin{theorem}\label{thm:tannercorrection}
Let $\innercode$ be a linear code over $\Sigma$ of length $d$ and rate $\innerrate > 1/2$.
Suppose that $\innercode$ has a $\innersmoothness$-smooth local reconstruction procedure with query complexity $\innerquery$.
Let $\outercode = \outercode_N(\innercode, G)$ be the expander code of length $N$ arising from the inner code $\innercode$ and a Ramanujan graph $G$.
Choose any $\gamma < 1/2$ and any $\zeta > \gamma$ satisfying
$\gamma \inparen{ e^\zeta \innerquery }^{-1/\gamma} > 8\lambda.$
Then $\outercode$ is $(q,\rho)$-locally correctable, 
for any error rate $\errorrate$, with $\errorrate < \gamma \inparen{ e^\zeta \innerquery }^{-1/\gamma} - 2 \lambda$.
The success probability is
\[ 1 - \inparen{ \frac{N}{d} }^{-1/\ln(d/4)} \]
and the query complexity is 
\[ \outerquery = \inparen{ \frac{N}{d} }^{ \eps } \qquad \text{ where } \qquad
 \eps = \inparen{ 1 + \frac{ \ln(\updatedquery) + 1 }{\zeta - \gamma} } \cdot \frac{ \ln(\updatedquery)}{\ln(d/4)}.\]
Further, when the length of the inner code, $d$, is constant, the correction algorithm runs in time $O( |\Sigma|^{\updatedquery + 1} \outerquery )$,
where $\updatedquery = \innerquery + (d - \innersmoothness)$.
\end{theorem}
\begin{remark} We will choose $d$ (and hence $\updatedquery < d$) and $|\Sigma|$ to be constant.  Thus, the rate of $\outercode$, as well as the parameters $\errorrate$ and $\eps$, will be constants independent of the block length $N$. 
The parameter $\zeta$ trades off between the query complexity and the allowable error rate.  
When $\innerquery$ is much smaller than $d$ (for example, $\innerquery = 3$ and $d$ is reasonably large), we will want to take $\zeta = O(1)$.  On the other hand, if $\innerquery = d^\eps$ and $d$ is chosen to be a sufficiently large constant, we should take $\zeta$ on the order of $\ln(\innerquery)$. 
\end{remark}

Before diving into the details, we outline the correction algorithm. 
First, we observe that it suffices to consider the case when $Q_0$ is perfectly smooth: that is, the queries of the inner code are uniformly random.
Otherwise, if $Q_0$ is $\innersmoothness$-smooth with $\innerquery$ queries, we may modify it so that it is $d$-smooth with $\innerquery + (d - \innersmoothness)$ queries, by having it query extra points and then ignore them.  
Thus, we set $\updatedquery = \innerquery$ and assume in the following that $Q_0$ makes $\innerquery$ perfectly smooth queries.

Suppose that $\innercode$ has local reconstruction
algorithm $(Q_0,A_0)$, and we receive a corrupted codeword, $w$, which differs from a correct codeword $c^*$ in at most a $\errorrate$ fraction of the entries.  Say we wish to determine $c^*[(u_0,v_1)]$, for $(u_0,v_1) \in E(H)$. 
The algorithm proceeds in two steps.  The first step is to find a set of about $N^{\eps/2}$ query positions which are nearly uniform in $[N]$, and whose correct values together determine $c^*[(u_0,v_1)]$.  The second step is to correct each of \em these \em queries with very high probability---for each, we will make another $N^{\eps/2}$ or so queries.
\paragraph{Step 1.} 
By construction, $c^*[(u_0,v_1)]$ is a symbol in a codeword of the inner code,
$\innercode$, which lies on the edges emanating from $u_0$.  By applying
$Q_0$, we may choose $\innerquery$ of these edges, $S = \inset{ (u_0,s_1^{(i)})
\suchthat i \in [\innerquery] }$, so that 
\[ A_0 \inparen{ \left.c^*\right|_S, (u_0,v_1)} = c[(u_0,v_1)].\]
Now we repeat on each of these edges: each
$(u_0, s_1^{(i)})$ is part of a codeword emanating from $s_1^{(i)}$, and so
$\innerquery$ more queries determine each of those, and so on.  Repeating this
$L_1$ times yields a $\innerquery$-ary tree $T$ of depth $L_1$, whose nodes are
labeled by of edges of $H$.  This tree-making procedure is given more precisely
below in Algorithm \ref{algo:maketree}.  Because the queries are smooth, each
path down this tree is a random walk in $H$; because $G$ is an expander, this
means that the leaves themselves, while not independent, are each close to
uniform on $E(H)$. 
Note that at this point, we have not made any queries, merely documented a tree, $T$, of edges
we could query.

\paragraph{Step 2.}
Our next step is to actually make queries to determine the correct values on the edges represented in the leaves of $T$.
By construction, these values determine $c^*[(u_0,v_1)]$.
Unfortunately, in expectation a $\errorrate$ fraction of the leaves are corrupted, and without further constraints on $\innercode$, even one corrupted leaf is enough to give the wrong answer.  To make sure that we get all of the leaves correct, we use the fact that each leaf corresponds to a position in the codeword that is nearly uniform (and in particular nearly independent of the location we are trying to reconstruct).   For each edge, $e$, of $H$ that shows up on a leaf of $T$, we repeat the tree-making process beginning at this edge, resulting in new $\innerquery$-ary trees $T_e$ of depth $L_2$.  This time, we make all the queries along the way, resulting in an evaluated tree $\tau_e$, whose nodes are labeled by elements of $\Sigma$; the root of $\tau_e$ is the $e$-th position in the corrupted codeword, $w[e]$, and we hope to correct it to $c^*[e]$.

For a fixed edge, $e$, on a leaf of $T$, we will correct the root of $\tau = \tau_e$ with very high probability, large enough to tolerate a union bound over all the trees $\tau_e$.
For two labelings $\sigma$ and
$\nu$ of the same tree by elements of $\Sigma$, we define the distance
\begin{equation}\label{eq:pathdist}
D(\sigma, \nu) = \max_P \Delta\inparen{ \left. \sigma \right|_{P}, \left. \nu \right|_{P} },
\end{equation}
where the maximum is over all paths $P$ from the root to a leaf,  and $\left. \sigma \right|_P$ denotes the restriction of $\sigma$ to $P$.
We will show below in Section \ref{sec:correctness} that it is very unlikely
that $\tau$ contains a path from the root to a leaf with more than a
constant fraction $\gamma < 1/2$ of errors.   Thus,
in the favorable case, the distance between the correct tree $\tau^*$ arising from $c^*$ and the observed tree $\tau$ is at most $D(\tau^*,\tau) \leq \gamma$.  In contrast, we will show that if $\sigma^*$ and $\tau^*$ are both trees arising from legitimate codewords with distinct roots, then $\sigma^*$ and $\tau^*$ must differ on an entire path $P$, and so $D(\sigma^*, \tau) > 1 - \gamma$.  
To take advantage of this, we show in 
Algorithm \ref{algo:round} how to efficiently compute
\[ \textsf{Score}(a) = \min_{\sigma^* : \text{root}(\sigma^*) = a} D(\sigma^*, \tau) \]
for all $a$, where $\text{root}(\sigma^*)$ denotes the label on the root of $\sigma^*$.  The above argument (made precise below in Section \ref{sec:correctness}) shows that there will be a unique $a \in \Sigma$ with score less than $\gamma$, and this will be the correct symbol $c^*[e]$. 

Finally, with all of the leaves of $T$ correctly evaluated, we may use $A_0$ to work our way back up $T$ and determine the correct symbol corresponding to the edge at the root of $T$.
The complete correction algorithm is given below in Algorithm \ref{algo:correct}.
\begin{algorithm}[!ht]
\DontPrintSemicolon
\KwIn{An index $e_0 \in E(H)$, and a corrupted codeword $w \in \Sigma^{E(H)}$.}
\KwOut{With high probability, the correct value of the $e_0$'th symbol.}
Set $L_1 = \log(n)/\log(d/4)$ and fix a parameter $L_2$\;
$T =$ \textsf{makeTree}$( e_0, L_1 )$\;
\For{each edge $e$ of $H$ that showed up on a leaf of $T$}
{
	$T_e =$ \textsf{makeTree}$(e, L_2)$\;
	Let $\tau_e = \left. T_e \right|_{w}$ be the tree of symbols from $w$\;
	$w^*[e] = \text{\textsf{correctSubtree}}(\tau_e)$\;
}
Initialize a $\innerquery$-ary tree $\tau^*$ of depth $L_1$\;
Label the leaves of $\tau^*$ according to $T$ and $w^*$: if a leaf of $T$ is labeled $e$, label the corresponding leaf of $\tau^*$ with $w^*[e]$.\; 
Use the local reconstruction algorithm $A_0$ of $\innercode$ to label all the nodes in $\tau^*$\;
\Return The label on the root of $\tau^*$
\caption{\textsf{correct}: Local correcting protocol.}
\label{algo:correct}
\end{algorithm}
\begin{algorithm}[!ht]
\DontPrintSemicolon
\KwIn{An initial edge $e_0 = (u_0,v_1) \in E(H)$, and a depth $L$.}
\KwOut{A $\innerquery$-ary tree $T$ of depth $L$, whose nodes are indexed by edges of $H$, with root $e_0$}

Initialize a tree $T$ with a single node labeled $e_0$\; 
$s = 0$\;
\For{ $\ell \in [L]$}{
	Let \textsf{leaves} be the current leaves of $T$\;
	\For{$e = (u_s,v_{1 - s}) \in $ \textsf{leaves}}{
		Let $\inset{ v_{1-s}^{(i)} \suchthat i \in [d] }$ be the neighbors of $u_s$ in $H$\;
		Choose queries $Q_0(e) \subset \inset{ (u_s, v_{1-s}^{(i)}) \suchthat i \in [d]}$, and add each query in $T$ as a child at $e$.\;
	}
	$s = 1 - s$\;
}
\Return $T$\;
\caption{\textsf{makeTree}: Uses the local correction property of $\innercode$ to construct a tree of indices.}
\label{algo:maketree}
\end{algorithm}
\begin{algorithm}[!ht]
\DontPrintSemicolon
\KwIn{$\tau$, a $\innerquery$-ary tree of depth $L$ whose nodes are labeled with elements of $\Sigma$.}
\KwOut{A guess at the root of the correct tree $\tau$.}
For a node $x$ of $\tau$, let $\tau[x]$ denote the label on $x$.\;
\For{ leaves $x$ of $\tau$ and $a \in \Sigma$}
{
	$\textsf{best}_a(x) = \begin{cases} 1 & \tau[x] \neq a \\ 0 & \tau[x] = a \end{cases}$\;
}
\For{$\ell = L-1, L-2, \ldots, 0$}
{
	\For{ nodes $x$ at level $\ell$ in $\tau$ and $a \in \Sigma$ }
	{
		Let $y_1, \ldots, y_{\innerquery}$ be the children of $x$\;
		Let $S_a \subset \Sigma^{\innerquery}$ be the set of query responses for the children of $x$ so that $A_0$ returns $a$ on those responses\;
		$\textsf{best}_a(x) = \min_{(a_0,\ldots,a_{\innerquery}) \in S_a} \max_{r \in [\innerquery]} \inparen{ \textsf{best}_{a_r}(y_r) + \ind{ \tau(y_r) \neq a_r}}$\;
	}
}
Let $r$ be the root of $\tau$\;
\For{ $a \in \Sigma$ }
{
	\[\textsf{Score}(a) = \frac{\textsf{best}_a(r) + \ind{\tau(r) \neq a}}{L} \]\;
}
\Return{ $a \in \Sigma$ with the smallest \textsf{Score}$(a)$}

\caption{\textsf{correctSubtree}: Correct the root of a fully evaluated tree $\tau$.}
\label{algo:round}
\end{algorithm}

The number of queries made by Algorithm \ref{algo:correct} is 
\begin{equation}\label{eq:nqueries} 
\outerquery = \innerquery^{L_1 + L_2}
\end{equation}
and the running time is $O( t_d |\Sigma|^{\innerquery + 1} \outerquery )$, where $t_d$ is the time required to run the local correction algorithm of $\innercode$.  For us, both $d$ and $|\Sigma|$ will be constant, and so the running time is $O(\outerquery)$.

\subsection{Proof of Theorem \ref{thm:tannercorrection}}\label{sec:correctness}
Suppose that $c^* \in \outercode$, and Algorithm \ref{algo:correct} is run on a received word $w$ with $\Delta(c^*,w) \leq \rho$.  To prove Theorem \ref{thm:tannercorrection}, we must show that Algorithm \ref{algo:correct} returns $c^*[e_0]$ with high probability.
As remarked above, we assume that $Q_0$ is perfectly smooth.

We follow the proof outline sketched in Section \ref{sec:tannerdecoding}, which rests on the following observation.
\begin{proposition}\label{prop:badpath}
Let $c_1,c_2 \in \outercode$ and let $e \in E(H)$ so that $c_1[e] \neq c_2[e]$.
Let the distance $D$ between trees with labels in $\Sigma$ be as in \eqref{eq:pathdist}.  
Let $T = \textsf{makeTree}(e)$, and let $\tau = \left. T\right|_{c_1}$ and $\sigma = \left. T\right|_{c_2}$ be the 
labeled trees corresponding to $c_1$ and $c_2$ respectively.
Then
$ D(\tau, \sigma) = 1.$
That is, there is some path from the root to the leaf of $T$ so that $\tau$ and $\sigma$ disagree on the entire path.
\end{proposition}
\begin{proof}
Since $c_1[e] \neq c_2[e]$, $\tau$ and $\sigma$ have different symbols at their root.  Since the labels on the children of any node determine the label on the node itself (via the local correction algorithm), it must be that $\tau$ and $\sigma$ differ on some child of the root.  Repeating the argument proves the claim.
\end{proof}
Let $\tau_e$ be the tree arising from the received word $w$, starting at $e$, as in Algorithm \ref{algo:correct}.  
Let
\[\mathcal{T}_e = \inset{ \left.\textsf{makeTree}(e)\right|_{c} \suchthat c \in \outercode}\]
be the set of query trees arising from uncorrupted codewords, and let $\tau^*_e \in \mathcal{T}_e$ be the ``correct" tree, corresponding to the original uncorrupted codeword $c^*$.  Suppose that 
\begin{equation}\label{eq:need}
D(\tau_e, \tau^*_e) \leq \gamma
\end{equation}
for some $\gamma \in [0,1/2)$.
Then Proposition \ref{prop:badpath} implies that for any $\sigma^*_e \in \mathcal{T}_e$ with a different root from $\tau^*_e$ has
\begin{equation}\label{eq:needimplies}
 D(\tau_e, \sigma^*_e) \geq 1 - \gamma.
\end{equation}
Indeed, there is some path along which $\tau^*_e$ and $\sigma^*_e$ differ in every place, and along this path, $\tau_e$ agrees with $\tau^*_e$ in at least a $1 - \gamma$ fraction of the places.  Thus, $\tau_e$ disagrees with $\sigma^*_e$ in those same places, establishing \eqref{eq:needimplies}.
Consider the quantity
\begin{equation}\label{eq:scoredef}
	 \textsf{Score}(a) = \min_{\sigma^*_e \in \mathcal{T}_e : \text{root}(\sigma^*_e) = a} D(\tau_e, \sigma^*_e).
\end{equation}
Equations \eqref{eq:need} and \eqref{eq:needimplies} imply that if $a^*$ is the label on the root of $\tau_e^*$, then $\textsf{Score}(a) \leq \gamma$, and otherwise, $\textsf{Score}(a) \geq 1 - \gamma$.  Thus, to establish the correctness of Algorithm \ref{algo:correct}, it suffices to argue first that Algorithm \ref{algo:round} correctly computes \textsf{Score}$(a)$ for each $a$, and second that \eqref{eq:need} holds for all trees $\tau_e$ in Algorithm \ref{algo:correct}.

The first claim follows by inspection.  Indeed, for a node $x \in \tau_e$, let $\inparen{\tau_e}_x$ denote the subtree below $x$.  
Let $\mathcal{T}_e^{(x,a)}$ denote the set of trees in $\mathcal{T}_e$ so that the node $x$ is labeled $a$.
Throughout Algorithm \ref{algo:correct}, the quantity $\textsf{best}_a(x)$ gives the distance from 
the observed tree rooted at $x$ to the best tree in $\mathcal{T}_e$, rooted at $x$, with the additional restriction that the label at $x$ should be $a$.  
That is,
\begin{equation}\label{eq:best}
 \textsf{best}_a(x) = \min_{\sigma^*_e \in \mathcal{T}_e^{(x,a)}} \tilde{D}\inparen{\inparen{\sigma^*_e}_x, \inparen{\tau_e}_x},
\end{equation}
where $\tilde{D}$ is the same as $D$ except it does not count the root, and it is not normalized.
It is easy to see that \eqref{eq:best} is satisfied for leaves $x$ of $\tau_e$.  Then for each node, Algorithm \ref{algo:round} updates $\textsf{best}_a(x)$ by considering the best labeling on the children of $x$ consistent with $\tau(x) = a$, taking the distance of the worst of those children, and adding one if necessary.

To establish the second claim, that \eqref{eq:need} holds for all trees $\tau_e$, 
we will need the following lemma about random walks on $H$.
\begin{lemma}\label{lem:pathsaregood}
Let $G$ and $H$ be as above, and suppose $\errorrate > 6\lambda$.
Let $v_0, \ldots, v_L$ be a random walk of length $L$ on $H$, starting from the left side at a vertex chosen from a distribution $\nu$ with $\twonorm{\nu - \frac{1}{n}\mathbf{1}_{n}} \leq \frac{1}{\sqrt{n}}$.
Let $X$ denote the number of corrupted edges included in the walk, and let $\errorrate + 2\lambda < \gamma < 1/2$.  Then
\[ \PR{X \geq \gamma L} \leq \exp\inparen{-L\,D\inparen{\gamma || \errorrate + 2\lambda}}.\]
\end{lemma}

Lemma \ref{lem:pathsaregood} says that a random walk on $H$ will not hit too many corrupted edges, which is very much like the expander Chernoff bound~\cite{IK10,Gil98}.  In this case, $H$ is the double cover of an expander, not an expander itself, and the edges, rather than vertices, are corrupted, but the proof remains basically the same.  For completeness, we include the proof of Lemma \ref{lem:pathsaregood} in the appendix.
The conditions on $\rho$ and $\lambda$ in the statement of Theorem \ref{thm:tannercorrection} implies that $\rho > 6\lambda$, and so Lemma \ref{lem:pathsaregood} applies to random walks on $H$.

Suppose that $L_1$ is even, and consider any leaf of $T$.
This leaf has label $(u_0, v_1) \in E(H)$, where $u$ is the result of a random walk of length $L_1$ on $G$ and $v$ is a randomly chosen neighbor of $u$.  Because $G$ is a Ramanujan graph, the distribution $\mu$ on $u$ satisfies
\[ \twonorm{ \mu - \frac{1}{n} \mathbf{1}_n } \leq \lambda^{L_1} \leq \frac{1}{\sqrt{n}} \]
as long as 
\[ L_1 \geq \frac{\log(n)}{\log(d/4)}.\]
Thus, Lemma \ref{lem:pathsaregood} applies to random walks in $H$ starting at $e$.  Fix a leaf of $\tau_e$;
by the smoothness of the query algorithm $Q_0$, each path from the root to the leaf of each tree $\tau_e$ is a uniform random walk, and
so with high probability, the number of corrupted edges on this walk is not more than $\gamma L_2$, which was the desired outcome.
The failure probability guaranteed by Lemma~\ref{lem:pathsaregood} is at most
\begin{align*}
	\exp( -L_2 D(\gamma || \errorrate + 2\lambda ) ) &= \inparen{ \frac{ \errorrate + 2\lambda }{\gamma}}^{\gamma L_2} \inparen{ \frac{1 - \errorrate - 2\lambda }{ 1 - \gamma }}^{(1 - \gamma)L_2 } \\
&\leq (e^\zeta \innerquery)^{-L_2} \inparen{ \frac{1}{1 - \gamma} }^{(1 - \gamma)L_2} \\
&\leq (e^\zeta \innerquery)^{-L_2} e^{\gamma L_2}.
\end{align*}
Above, we used the assumption that $\errorrate + 2\lambda < \gamma \inparen{ e^\zeta \innerquery }^{-1/\gamma}$ from the statement of Theorem~\ref{thm:tannercorrection}.

Finally, we union bound over $\innerquery^{L_1}$ trees $\tau_e$ and $\innerquery^{L_2}$ paths in each tree. 
We will set $L_2 = CL_1$, for a constant $C$ to be determined.
Thus, \eqref{eq:need} holds (and hence Algorithm \ref{algo:correct} is correct) except with probability at most
\begin{align} 
	\PR{ \text{Algorithm \ref{algo:correct} fails} } 	&\leq \innerquery^{L_1+L_2} \inparen{ e^\zeta \innerquery }^{-L_2} e^{\gamma L_2 }  \nonumber\\
														&= \exp\inparen{ (C + 1)L_1 \ln(\innerquery) - C L_1(\zeta + \ln(q_0)) + C\gamma L_1 }.
\end{align}

Our goal is to show that $\PR{ \text{Algorithm \ref{algo:correct} fails} } \leq \exp(-L_1)$, which is equivalent to showing 
\[
(C + 1) \ln(\innerquery) - C (\zeta + \ln(q_0) ) + C\gamma < -1.
\]
This holds if we choose
\[ C < \frac{ 1 + \ln(q_0) }{\zeta - \gamma}. \]
From \eqref{eq:nqueries}, $q = \innerquery^{(C+1)L_1}$, which completes the proof of Theorem \ref{thm:tannercorrection}.

\section{Examples}\label{sec:tannerexamples}

In this section, we provide two examples of choices for $\innercode$, both of which result in $(N^\eps, \errorrate)$-LCCs of rate $1 - \alpha$ for any constants $\eps, \alpha > 0$ and for some constant $\errorrate > 0$.  Our first and main example is a generalization of Reed-Muller codes, based on finite geometries.  With these codes as $\innercode$, we provide LCCs over $\F_p$---unlike multiplicity codes, these codes work naturally over small fields.

Our second example comes from the observation that if the $\innercode$ is itself an LCC (of a fixed length) our construction provides a new family of $(N^\eps,\errorrate)$-LCCs.
In particular, plugging the multiplicity codes of~\cite{KSY11} into our construction yields a novel family of LCCs.
This new family of LCCs has a very different structure than the underlying multiplicity codes, but achieves roughly the same rate and locality.

\paragraph{Codes from Affine Geometries.}
One advantage of our construction is that the inner code $\innercode$ need not
actually be a good locally decodable or correctable code.  Rather, we only need a smooth
reconstruction procedure, which is easier to come by.  One example comes from affine geometries; in this example, we will show how use Theorem \ref{thm:tannercorrection} to make LCCs of length $N$, rate $1 - \alpha$ and query complexity $N^\eps$, for any $\alpha,\eps > 0$.
\newcommand{\fsize}{h}  

For a prime power $\fsize = p^\ell$ and parameters $r$ and $m$, consider the $r$-dimensional affine subspaces $L_1, \ldots, L_t$ of the vector space $\F_{\fsize}^m$.  
let $H$ be the $t \times \fsize^m$ incidence matrix of the $L_i$ and the points of $\F_{\fsize}^m$, and let $\mathcal{A}^*(r,m,\fsize)$ be the code over $\F_p$ whose parity check matrix is $H$.
These codes, examples of \em finite geometry codes, \em are well-studied, and their ranks can be exactly computed---see~\cite{assmus1994,assmus1998} for an overview.

The definition of of $\mathcal{A}^*(r,m,\fsize)$ gives a reconstruction procedure: we may query all the points in a random $r$-dimensional affine subspace of $\F_{\fsize}^m$ and use the corresponding parity check.  
In particular, if we index the positions of the codeword by elements of $\F_{\fsize}^m$.  Then given the position $x \in \F_\fsize^m$, the query set $Q(x)$ is all the points other than $x$ in a random $r$-flat $L$ that passes through $x$.  
Given a codeword $c \in \mathcal{A}^*(r,m,\fsize)$, we may reconstruct $c_x$ by
\[ A\inparen{\left.c\right|_{Q(x)} } = -\sum_{y \in Q(x)} c_y.\]
By definition, $(A,Q)$ is a smooth reconstruction procedure which makes $\fsize^{r}$ queries.

The locality of $\mathcal{A}^*(r,m,\fsize)$ has been noticed before, for example in~\cite{GKS12}, where it was observed that these codes could be viewed as lifted parity check codes.  However, as they note, these codes do not themselves make good LCCs---the reconstruction procedure cannot tolerate any errors in the chosen subspace, and thus the error rate $\rho$ must tend to zero as the block length grows.
Even though these codes are not good LCCs, we can  use them in Theorem \ref{thm:tannercorrection} to obtain good LCCs with sublinear query complexity, which can correct a constant fraction of errors.
We will use the bound on the rate of $\mathcal{A}^*(1,m,\fsize)$ from~\cite{GKS12}:
\begin{lemma}[Lemma 3.7 in~\cite{GKS12}]\label{lem:rate}
Choose $\ell = \eps m$, with $\fsize = p^\ell$ as above.
The dimension of $\mathcal{A}^*(1, m, \fsize)$ is at least $\fsize^m - \fsize^{m(1 - \beta)}$, for $\beta = \beta(\eps') = \Omega(2^{-2/\eps'})$.
\end{lemma}
We will apply Lemma \ref{lem:rate} with 
\[ \eps' = \frac{\eps}{2} \qquad \text{and} \qquad  m = \sqrt{ \frac{\ln(2/\alpha)}{\eps' \beta(\eps')\ln(p)}},\]
to obtain a $p$-ary code $\innercode$ of length $d = p^{\eps'm^2}$ 
with rate $\innerrate$ at least $1 - \alpha/2$ and which has a $(d-1)$-smooth reconstruction algorithm with query complexity $\innerquery = d^{\eps'}$.
To apply Theorem \ref{thm:tannercorrection}, fix any $\eps,\alpha > 0$, sufficiently small.
We set $\zeta = 2\ln(\innerquery)$, and choose $\gamma = 1/4$ in Theorem \ref{thm:tannercorrection}, and use $\innercode$: the resulting expander code $\outercode$ has rate $1 - \alpha$ and query complexity
\[ q \leq \inparen{ \frac{N}{d} }^\eps \]
for sufficiently large $d$.  Finally, using the fact that $\lambda \leq 2/\sqrt{d}$, we see that $\outercode$ corrects against a $\rho$ fraction of errors, where
\[ \rho = \frac{1}{5} d^{ -6\eps'} \]
again for sufficiently large $d$, as long as $\eps < 1/12$.
Assuming $\eps$ and $\alpha$ are small enough that $d$ is a suitably large constant, this rate $\rho$ is a positive constant, and we achieve the advertised results.

\paragraph{Multiplicity codes.} 
Multiplicity codes \cite{KSY11} are themselves a family of constant-rate locally decodable codes.
We can, however, use a multiplicity code of constant length as the inner code $\innercode$ in our construction.
This results in a new family of constant-rate locally decodable codes.
The parameters we obtain from this construction are slightly worse than the original multiplicity codes, 
and the main reason we include this example is novelty---these new codes have a very different structure than the
original multiplicity codes. 

For constants
$\alpha', \eps' > 0$, the multiplicity codes of~\cite{KSY11} have length $d$ and rate $\innerrate = 1 -
\alpha'$ and a $(d-1)$-smooth local reconstruction algorithm with query complexity $\innerquery = O(d^{\eps'})$.   
To apply Theorem \ref{thm:tannercorrection},
we will choose $\zeta = C \ln(\innerquery)$ for a sufficiently large constant $C$, and so the query complexity of $\outercode$ will be
\[ \outerquery = \inparen{ \frac{N}{d} }^{(1 + \beta) \eps'}\]
for an arbitrarily small constant $\beta$.
Thus, setting $\eps = \eps'(1 + \beta)$, and $\alpha = 2\alpha'$, we obtain codes $\outercode$ with rate $1 - \eps$ and query complexity $(N/d)^\eps$.
As long as $\eps$ is sufficiently small, $\outercode$ can tolerate errors up to $\rho = C' d^{-C'' \eps}$ for constants $C'$ and $C''$ (depending on the constants in the constructions of the multiplicity code, as well as on $C$ above).
Multiplicity codes require sufficiently large block length $d$, on the order of
\[d \approx \inparen{ \frac{1}{\alpha^2 \eps^3 } }^{1/\eps} \log\inparen{\frac{1}{\alpha\eps}}.\]
Choosing this $d$ results in a requirement $\rho \leq 1/\poly(\alpha \eps)$.
We remark that the distance of the multiplicity codes is
on the order of $\innerdist = \Omega(\alpha^2\eps)$, and so the distance
of the resulting expander code $\outercode$ is $\Omega(\alpha^4
\eps^2)$. 

\section{Conclusion}
In the constant-rate regime, all known LDCs work by using a smooth local reconstruction algorithm.  When the locality is, say, three, then with very high probability none of the queried positions will be corrupted.  This reasoning fails for constant rate codes, which have larger query complexity: we expect a $\errorrate$ fraction of errors in our queries, and this is often difficult to deal with.  In this work, we have shown how to make the low-query argument valid in a high-rate setting---any code with large enough rate and with a good local reconstruction algorithm can be used to make a full-blown locally correctable code.

The payoff of our approach is the first sublinear time algorithm for locally correcting
expander codes.  More precisely, we have shown that as long as the inner code
$\innercode$ admits a smooth local reconstruction algorithm with appropriate
parameters, then the resulting expander code $\outercode$ is a $(N^\eps,
\errorrate)$-LCC with rate $1 - \alpha$, for any $\alpha, \eps > 0$ and some
constant $\rho$.  
Further, we presented a decoding algorithm with runtime linear in the number of queries.

There are only two other constructions known in this regime, and 
and our constructions are substantially different.  
Expander codes are a natural construction, and it is our hope that
the additional structure of our codes, as well as the extremely fast decoding
time, will lead to new applications of local decodability.

\bibliographystyle{alpha}
\bibliography{refs.bib}

\appendix

\section{Proof of Lemma \ref{lem:pathsaregood}}
In this appendix, we provide a proof of Lemma \ref{lem:pathsaregood}.
The lemma follows with only a few tweaks from standard results. 
The only differences between this and a standard analysis of random walks on expander graphs are that (a) we are walking on the edges of the bipartite graph $H$, rather than on the vertices of $G$, and (b) our starting distribution is not uniform but instead close to uniform.
 Dealing with this differences is straightforward, but we document it below for completeness.

First, we need the relationship between a walk on the edges of a bipartite graph $H$ and the corresponding walk on the vertices of $G$.  
For ease of analysis, we will treat $H$ as directed, with one copy of each edge in each direction.

\begin{lemma}\label{lem:walk}
Let $G$ be a degree $d$ undirected graph on $d$ vertices with normalized adjacency matrix $A$, and let $H$ be the double cover of $G$.
For each vertex $v$ of $G$, label the edges incident to $v$ arbitrarily, and let $v(i)$ denote the $i^{th}$ edge of $v$.
Let $H'$ be the graph with vertices $V(G) \times [d] \times \{0,1\}$ and edges 
\[ E(H') = \inset{ ((u,i,b), (v,j,b')) \suchthat (u,v) \in E(G), b \neq b', u(i) = v }.\]
Then $H'$ is a directed graph with $2dn$ edges, and in-degree and out-degree both equal to $d$.
Further, the normalized adjacency matrix $A'$ is given by
\[ A' = R \otimes S \]
where $S:\R^2 \to \R^2$ is $S = \begin{bmatrix} 0& 1\\1&0 \end{bmatrix}$ and $R:\R^{nd} \to \R^{nd}$ is an operator with the same rank and spectrum as $A$. 
\end{lemma}

\begin{proof}
We will write down $A'$ in terms of $A$.  Index $[n]$ by vertices of $V$, so that $e_v \in \R^n$ refers to the standard basis vector with support on $v$. Let $\otimes$ denote the Kronecker product.  We will need some linear operators.  Let $B: \R^{n^2} \to \R^{n^2}$ so that
\[B(e_u \otimes e_v) = e_v \otimes e_v\]
and $P:\R^{n^2} \to \R^{nd}$ so that
\[ P(e_u \otimes e_v) = \begin{cases} e_u \otimes e_i & v = u(i) \\ 0 & (u,v) \not\in E(G) \end{cases} .\]
Finally, let $S:\R^2 \to \R^2$ be the cyclic shift operator.
Then a computation shows that the adjacency matrix $A'$ of $H'$ is given by
\[ (P(I \tensor A) B P^T) \tensor S.\]
Let $R = P(I \tensor A) BP^T$.  To see that the rank of $R$ is at most $n$, note that for 
any $i \in [d]$ and any $u \in V(G)$,
\[ R(e_u \otimes e_j) = e_{u(j)} \otimes \frac{1}{d} \mathbf{1}_d.\]
In particular, it does not depend on the choice of $j$.  Since $\inset{ e_u \otimes e_j \suchthat u \in V(G), j \in [d] }$ is a basis for $\R^{nd}$, the image of $R$ has dimension at most $n$. 
Finally, a similar computation shows that if $p$ is an eigenvector of $A$ with eigenvalue $\lambda$, then $p \otimes \frac{1}{d}\mathbf{1}_d$ is a right eigenvector of $R$, also with eigenvalue $\lambda$.  (The left eigenvectors are $P(\frac{1}{n}\mathbf{1}_n \otimes p)$).
This proves the claim.
\end{proof}

With a characterization of $A'$ in hand, we now wish to apply an expander Chernoff bound.  Existing bounds require slight modification for this case (since the graph $H'$ is directed and also not itself an expander), so for completeness we sketch the changes required.  The proof below follows the strategies in~\cite{AFWZ95} and~\cite{IK10}.
We begin with the following lemma, following from the analysis of~\cite{AFWZ95}.
\begin{lemma}
Let $G$ and $H$ be as in Lemma \ref{lem:walk}, and
let $v_0, v_1, \ldots, v_T$ be a random walk on the vertices of $H$, beginning at a vertex of $H$, chosen as follows: the side of $H$ is chosen according to a distribution $\sigma_0 = (s, 1 - s)$, and the vertex within that side is chosen independently according to a distribution $\nu$ with $\| \nu - \frac{1}{n}\mathbf{1}_n\|_2 \leq \frac{1}{\sqrt{n}}$.  Let $W$ be any set of edges in $H$, with $|W| \leq \errorrate nd$.   Suppose that $\errorrate > 6\lambda$.  Then for any set $S \subset \{0,1,\ldots,T-1\}$, 
\[ \PR{ (v_t,v_{t+1}) \in W ,\forall t \in S} \leq (\errorrate + 2\lambda)^{|S|}.\]
\end{lemma}

\begin{proof}
As in Lemma \ref{lem:walk}, we will consider $H$ as directed, with one edge in each direction.  As before, we will index these edges by triples $(u,i,\ell) \in V(G) \times [d] \times \{0,1\}$, so that $(u,i,\ell)$ refers to the $i^{th}$ edge leaving vertex $u$ on the $\ell^{th}$ side of $H$. 
Let $\mu$ be the distribution on the first step $(v_0,v_1)$ of the walk, so
\[ \mu = \nu \otimes \frac{1}{d} \mathbf{1}_d \otimes \sigma_0.\]  

Let $M \in \R^{2nd}$ be the projector onto the edges in $W$.  Let $M^{(0)}$ be the restriction to edges emanating from the left side of $H$, and $M^{(1)}$ from the right side, so that both $M^{(0)}$ and $M^{(1)}$ are $nd \times nd$ binary diagonal matrices with at most $\errorrate nd$ nonzero entries.
Let $A' = R \otimes S$ be as in the conclusion of Lemma \ref{lem:walk}.  After running the random walk for $T$ steps, consider the
distribution on directed edges of $H$, conditional on the bad event that $(v_t,v_{t+1}) \in W$ for all $t \in S$.
As in the analysis in~\cite{AFWZ95}, this distribution is given by
\[ \mu_T = \frac{ (M_{T_1} A') (M_{T-2} A') \cdots (M_1 A')(M_0 \mu) }{ \PR{ (v_t, v_{t+1}) \in W, \forall t \in S}},\]
where
\[ M_t = \begin{cases} M & t \in S \\ I & t\not\in S \end{cases}.\]
Since the $\ell_1$ norm of any distribution is $1$, we have
\begin{align}
\PR{ (v_t,v_{t+1}) \in W, \forall t \in S} &= \onenorm{ (M_{T-1} A') (M_{T-2} A') \cdots (M_1 A')(M_0 \mu) } 
\label{eq:probeq}
\end{align}
Let
\[ \mu_0 := M_0 \mu,\]
and 
\[ \mu_t := M_t A' \mu_{t-1},\]
so we seek an estimate on $\onenorm{\mu_T}$.  

The following claim will be sufficient to prove the theorem.
\begin{claim} If $\errorrate \geq 6\lambda$, and $t \in S$,
\[ (\mu - 2\lambda)\onenorm{\mu_t} \leq \onenorm{\mu_{t+1}} \leq (\mu + 2\lambda)\onenorm{\mu_t}.\]
On the other hand, if $t \not\in S$,
\[ \onenorm{\mu_t} = \onenorm{\mu_{t+1}}. \]
\end{claim}
The second half of the claim follows immediately from the definition of $\mu_t$.
To prove the first half, suppose that $t \in S$.  We will proceed by induction.
Again, we follow the analysis of~\cite{AFWZ95}.

Write $\mu_0 = v_0 \otimes \sigma_0$, and write $\sigma_0 = (s,1 - s)$  Part of our inductive hypothesis will be that for all $t$,
\[ \mu_t = v_t^{(0)} \otimes s_t e_0 + v_t^{(1)} \otimes (1 - s_t) e_1,\]
where $s_t = s$ if $t$ is even and $1 - s$ if $t$ is odd,
and where $v_t^{(i)} \in \R^{nd}$. 
For $i \in \{0,1\}$, write
\[ v^{(i)}_t = x^{(i)}_t + y^{(i)}_t, \]
where $x^{(i)}_t \| \mathbf{1}$ and $y^{(i)}_t \perp \mathbf{1}$.  
The second part of the inductive hypothesis will be
\begin{equation}\label{eq:induct}
\|y^{(i)}_t\|_2 \leq q \|x^{(i)}_t\|_2,
\end{equation}
for a parameter $q$ to be chosen later, and for $i \in \{0,1\}$.

Because
\begin{align*}
\|\mu_t\|_1 &= s_t \|v^{(0)}_t\|_1 + (1 - s_t) \|v^{(1)}_t\|_1 \\
&= s_t \|x^{(0)}_t\|_1 + (1 - s_t) \|x^{(1)}_t\|_1\\
&=\sqrt{nd} \inparen{ s_t \|x^{(0)}_t\|_2 + (1 - s_t) \|x^{(1)}_t\|_2}, 
\end{align*}
it suffices to show that
\begin{equation}\label{eq:suffices}
 (\mu - 2\lambda) \twonorm{x^{(0)}_t} \leq \twonorm{x^{(1)}_{t+1}} \leq (\mu + 2\lambda) \twonorm{x^{(0)}_t}
\end{equation}
and similarly with the $0$ and $1$ switched.  The analysis is the same for the two cases, so we just establish \eqref{eq:suffices}.
Using the decomposition $A' = R \otimes S$ from Lemma \ref{lem:walk}, 
\begin{align*}
\mu_{t+1} &= M_t (R \otimes S)( v_t^{(0)} \otimes s_t e_0 + v_t^{(1)} \otimes (1 -s_t)e_1 )\\
&= M_t\inparen{ Rv_t^{(0)} \otimes (1 - s_{t+1})e_1 + Rv_t^{(1)} \otimes s_{t+1} e_0}\\
&= \inparen{M_t^{(1)} Rv_t^{(0)}} \otimes (1 - s_{t+1})e_1 + \inparen{M_t^{(0)} Rv_t^{(1)}} \otimes s_{t+1} e_0
\end{align*}
This establishes the first inductive claim about the structure of $\mu_{t+1}$, and
\[ v_{t+1}^{(0)} = M_t^{(0)} Rv_t^{(1)} \qquad \text{and} \qquad v_{t+1}^{(1)} = M_t^{(1)} Rv_t^{(0)}.\]
Consider just $v_{t+1}^{(1)}$.  We have
\[ v_{t+1}^{(1)} = M_t^{(1)} R (x_t^{(0)} + y_t^{(0)}).\]
Because $t \in S$, we know that $M_t^{(1)}$ is diagonal with at most $\errorrate nd$ nonzeros, and further we know that $R$ has second normalized eigenvalue at most $\lambda$, by Lemma \ref{lem:walk}.  
The analysis in~\cite{AFWZ95} now shows that, using the inductive hypothesis \eqref{eq:induct},
\begin{equation}\label{eq:sammich}
 \errorrate \|x_t^{(0)}\|_2 - q\lambda \sqrt{\errorrate(1 - \errorrate)} \|x_t^{(0)}\|_2 \leq \|x_{t+1}^{(1)}\|_2 \leq 
\errorrate \|x_t^{(0)}\|_2 + q\lambda \sqrt{\errorrate(1 - \errorrate)} \|x_t^{(0)}\|_2,
\end{equation}
and that
\[ \| y_{t+1}^{(1)}\|_2 \leq q\lambda \|x_t^{(0)}\|_2 + \sqrt{\errorrate(1 - \errorrate)} \|x_t^{(0)}\|_2.\]
We must ensure that \eqref{eq:induct} is satisfied for the next round.  As long as $\lambda < \errorrate/6$, this follows from the above when
\[ q = 2\sqrt{ \frac{1 - \errorrate}{\errorrate}}.\]
With this choice of $q$, the \eqref{eq:suffices} follows from \eqref{eq:sammich}.
Further, the hypotheses on $\nu$ show that the \eqref{eq:induct} is satisfied in the initial step. 
\end{proof}

Finally, we invoke the following theorem, from~\cite{IK10}.
\begin{theorem}[Theorem 3.1 in \cite{IK10}]
Let $X_1, \ldots, X_L$ be binary random variables so that for all $S \subset [L]$, 
\[ \PR{ \bigwedge_{i\in S} X_i = 1 } \leq \delta^{|S|}.\]
Then for all $\gamma > \delta$, 
\[ \PR{ \sum_{i=1}^L X_i \geq \gamma L } \leq e^{-LD(\gamma || \delta)}.\]
\end{theorem}
Lemma \ref{lem:pathsaregood} follows immediately.

\end{document}